\documentclass[a4paper,11pt,reqno,
final
]{amsart}	

\usepackage{amsmath}
\usepackage{amssymb}
\usepackage{amsthm}
\usepackage{amsfonts}
\usepackage{natbib}   

\newtheorem{theorem}{Theorem}[section]
\newtheorem{proposition}[theorem]{Proposition}

\newtheorem{assumption}[theorem]{Assumption}
\newtheorem{remark}[theorem]{Remark}
\newtheorem{definition}[theorem]{Definition}
\newtheorem{example}[theorem]{Example}

\def\E{{\mathbb E}}
\def\R{{\mathbb R}}

\newcommand{\ud}{\, {\rm d}}

\begin{document}

\title[Who would invest only in the risk-free asset?]{Who would invest only in the risk-free asset?}

\author[N. Azevedo]{$\textrm{N. Azevedo}^1$}
\address[N. Azevedo]{Financial Stability Department, Banco de Portugal \& ESEIG, Polytechnic Institute of Porto, Portugal}
\email{nazevedo@iseg.utl.pt}
\thanks{$^1$ The opinions expressed in this article are those of the authors and do not necessarily coincide with those of Banco de Portugal or the Eurosystem.}
\author[D. Pinheiro]{D. Pinheiro}
\address[D. Pinheiro]{Department of Mathematics, Brooklyn College of the City University of New York, NY, USA}
\email{dpinheiro@brooklyn.cuny.edu}
\author[S.Z. Xanthopoulos]{S.Z. Xanthopoulos}
\address[S.Z. Xanthopoulos]{Department of Mathematics, University of the Aegean, Samos, Greece}
\email{sxantho@aegean.gr}
\author[A.N. Yannacopoulos]{A.N. Yannacopoulos}
\address[A.N. Yannacopoulos]{Department of Statistics, Athens University of Economics and Business, Athens, Greece}
\email{ayannaco@aueb.gr}
\date{}

\begin{abstract}
Within the setup of continuous-time semimartingale financial markets, we show that a multiprior Gilboa-Schmeidler minimax expected utility maximizer forms a portfolio consisting only of the riskless asset if and only if among the investor's priors there exists a probability measure under which all admissible wealth processes are supermartingales. Furthermore, we show that under a certain attainability condition (which is always valid in finite or complete markets) this is also equivalent to the existence of an equivalent (local) martingale measure among the investor's priors. As an example, we generalize a no betting result due to Dow and Werlang.


\vspace{0.3cm}
\end{abstract}

\maketitle

\section{Introduction}

Expected utility maximization plays a prominent role in mathematical finance as a decision making tool. According to this paradigm, out of a family of wealth processes $(X_t(\pi))_{t\in[0,T]}$,  with regard to a market $S$,  and indexed in terms of portfolio processes $\pi$  in a certain set $\Pi$ of admissible portfolios, a decision maker (or investor) chooses the one corresponding to portfolio processes $\pi^{*}\in\Pi$ for which the following variational principle, known as \emph{portfolio optimization problem}, holds:
\begin{equation*}
{\mathbb E}_{P}[U(X_{T}(\pi^*))]=\sup_{\pi \in \Pi} {\mathbb E}_{P}[U(X_T(\pi))] \ .
\end{equation*}

Interestingly, the solvability of the portfolio optimization problem is closely related to qualitative and quantitative properties of the financial market described by $S$. In particular, it is well known that the well-posedness of the utility maximization problem is related to issues regarding the existence of an equivalent  (local) martingale measure (also known as risk neutral measure) and the corresponding availability of a linear pricing rule. This discussion has been initiated, for a variety of market models, in the seminal works of  \cite{Harrison1} and \cite{Harrison2,Harrison3}, where the concept of market viability is defined as the precise setting under which the above portfolio optimization problem has a solution in terms of a net trade compatible with certain constraints. It should also be remarked that the notion of market viability is closely related with absence of arbitrage, as well as with the existence of equivalent martingale measures, which can then be reinterpreted as appropriate pricing kernels as mentioned above. Indeed, the two notions are equivalent for finite markets. However, when considering continuous-time markets or models in infinite probability spaces, the situation becomes delicate on account of various mathematical intricacies.  While market viability remains robust as a concept, the notion of arbitrage, as well as this of equivalent martingale measure or pricing kernel, have to be refined. This contributed to the appearance of various alternative definitions, the connection among which was not always very clear, a fact that has led to the development of interesting literature by leading experts in the field.  In particular, various definitions have been introduced to express what should be meant by ``absence of arbitrage'', among which, just to mention a few, one should list those of \emph{no arbitrage} (NA), \emph{no unbounded profit with bounded risk} (NUPBR), \emph{no free lunch} (NFL), \emph{no free lunch with bounded risk} (NFLBR), and \emph{no free lunch with vanishing r
isk} (NFLVR). Links between these definitions have been thoroughly investigated. For instance, it is known that $NFL \Longrightarrow NFLBR \Longrightarrow NFLVR \Longrightarrow NUPBR$ and $NFLVR \Longleftrightarrow LUPBR + NA$. For a unified perspective of the most significant no arbitrage conditions, in the context of continuous semimartingale models, see \cite{fontana2015weak} and references therein. It shouldn't be surprising then that similar observations hold when discussing the existence of equivalent martingale measures. Indeed, each notion of absence of arbitrage is related to a generalized (often weaker) concept of equivalent martingale measure, each one of which having consequences on the behaviour of the portfolio optimization problem. For instance, as shown in \cite{karatzas2007numeraire} within a general semimartingale setting, the condition NUPBR is necessary for the solvability of the portfolio optimization problem, confirming and generalizing a similar result of \cite{lowenstein2000local}. Quite recently, \cite{choulli2015non} have shown that the condition NUPBR is equivalent to the solvability of the portfolio optimization problem (but possibly up to  an equivalent change of measure; see also \cite{christensen2007no} for a counter example). It is worth noting that these considerations have led to elegant formulations of the portfolio optimization problem in terms of duality methods, allowing for semi-explicit representations of the solutions (see e.g. \cite{kramkov1999asymptotic}, \cite{delbaen2006mathematics}).

However, the traditional Von Neumann-Morgenstern expected utility framework fails to address the problem of model uncertainty and ambiguity aversion that has been dictated by the famous Ellsberg paradox (\cite{ellsberg1961risk}), related with the distinction introduced in 1921 by Frank Knight between risk and uncertainty (\cite{knight1921uncertainty}).  According to this distinction, ``risk'' refers to the situation where the unique probability distribution of a random experiment is assumed to be known, while the term (Knightian) ``uncertainty'' is reserved for situations where such a unique probability assignment does not exist. It is precisely the fact that agents are not always capable of attaching a unique probability measure to the relevant state space that  was manifested by the Ellsberg paradox, reflecting the fact that when the information available is not ``sufficient'' to form a single probability distribution assumption, a decision maker may consider a whole set of alternative distributions as plausible models and then act on this consideration.   The model uncertainty and ambiguity aversion paradigm has been employed to provide some explanation of various empirical observations that did not comply to the more traditional expected utility model, such as for example the failure of the two-fund separation theorem, the equity premium and the risk-free rate puzzles, the trading freezes, etc..

There are two basic strands in the literature that extend the expected utility paradigm, to cope with model uncertainty. The first, introduced by \cite{schmeidler1989subjective}, is based on the use of non-additive  probabilities (capacities) to represent the decision maker's beliefs, while the second, introduced by \cite{gilboa1989maxmin}, allows for beliefs to be represented by a set of probabilities, while preferences are expressed by  the  ``maxmin''  on  the  set  of  expected  utilities. Under the Gilboa-Schmeidler approach, the traditional utility is replaced by the robust utility 
\begin{equation*}
{\mathcal U}(X):=\inf_{P\in{\mathcal P}}\E_P[U(X)] \ , 
\end{equation*} 
where ${\mathcal P}$ is a relevant set of priors concerning the distribution of  $X$. Now, the aim of the investor is to maximize this robust utility over all admissible trading strategies, leading to an optimization problem of the form 
\begin{equation*}
\sup_X\inf_{P\in{\mathcal P}}\E_P[U(X_T)] \ ,
\end{equation*}
where $X$ takes values in the set of all wealth processes associated with admissible portfolios. Knightian decision theory plays nowadays a prominent role in economic theory  (see e.g. \cite{bewley2002knightian} for a rigorous formulation of Knight's ideas), and finance in particular, as ambiguity introduces interesting effects in the portfolio optimization problem. For example, in the context of a simple one period model with a single asset,  \cite{dow1992uncertainty} have shown that, although within the expected utility paradigm trades occur generically, ambiguity may generate no betting intervals. This no betting effect is further addressed and reconfirmed  in later studies, see e.g. \cite{billot2000sharing}, \cite{easley2009ambiguity} and \cite{guidolin2010simple}. For a detailed recent review of the literature see \cite{guidolin2013ambiguity} and references therein.

Furthermore, there has been some recent interesting academic activity with regard to the characterization of the non arbitrage condition under model uncertainty (e.g.  \cite{bayraktar2015arbitrage}, \cite{bouchard2015arbitrage}, \cite{biagini2015robust}). For example, in \cite{bayraktar2015arbitrage} it is shown in a discrete time non-dominated model uncertainty setting that NA holds if and only if there exists a family of probability measures such that any admissible value process is a local supermartingale under these measures.

The goal of this work is to  extend the discussion concerning the connection between the existence of martingale and supermartingale measures and the portfolio optimization problem in the framework of minimax utility. In particular, we show that for a general class of minimax utilities of the Gilboa-Schmeidler type, the optimal portfolio consists purely of investment in the riskless asset if and only if among the investor's priors there exists a probability measure under which all admissible wealth processes are supermartingales. In addition, we show that under a certain attainability condition (which is always valid in finite or complete markets) this is also equivalent to the existence of an equivalent (local) martingale measure among the investor's priors.  Furthermore, we show in a simple example a potentially interesting  link of our results to the  no betting or ``market freezes'' phenomenon as described in \cite{dow1992uncertainty}.

This paper is organized as follows. In Section \ref{Sec2} we describe the setting we work with. Section \ref{Sec3} is devoted to the analysis of the particular, yet very relevant, case of Von Neuman-Morgernstern utilities and in Section \ref{Sec4} we state and prove our main results.

\section{Setup and problem formulation}\label{Sec2}

We consider a securities market consisting of $d+1$ assets, one riskless and $d$ risky assets. The riskless asset is assumed to bear a deterministic instantaneous risk free interest rate.  Then, without loss of generality and to ease notation, we may assume that the rate of return of the riskless asset is $r=0$. Otherwise, we may simply use the (non-zero) price of the riskless asset as num\'eraire. Let $T>0$ be a fixed finite horizon. The risky assets price process will be denoted by $S=(S_t)_{t\in[0,T]}$ with $S_t=\left(S^{1}_t,\ldots, S^{d}_t\right)$, $t \in [0,T]$, where the coordinate processes $S^{i}$, $i=1,2,\ldots,d$, are assumed to be semimartingales on a filtered probability space $(\Omega, {\mathcal F}, ({\mathcal F}_{t})_{t\in[0,T]},P)$, for any  probability measure $P$ under consideration.

A portfolio $\pi\in\Pi$ is a pair $(x,H)$ where the constant $x$ is the initial wealth and $H=(H_t)_{t\in[0,T]}$ with $H_t=\left(H^{1}_t,\ldots, H^{d}_t\right)$, $t \in [0,T]$, is a predictable process specifying the amount of risky assets held in the portfolio. The wealth of an investor with a self financing portfolio $\pi=(x,H)$ is given by the stochastic process $X=(X_t)_{t\in[0,T]}$ defined as
\begin{eqnarray*}
X_t=x + \int_{0}^{t} H_u \ud S_u
\end{eqnarray*}
for every $t\in[0,T]$.

On the filtered measurable space  $(\Omega, {\mathcal F}, ({\mathcal F}_{t})_{t\in[0,T]})$  we assign  a  set of priors ${\mathcal P}$ consisting of the subjective probability laws that may govern the market according to the beliefs of the investor.  Decisions concerning optimal portfolio choice are made using a minimax utility framework, to accommodate the effects of robustness and model uncertainty. Hence, each investor solves a minimax problem of the form
\begin{eqnarray}\label{MIN-MAX-A-A}
\sup_{X \in {\mathcal X}(x)} \inf_{P \in {\mathcal P}} {\mathbb E}_{P}[U(X_T)].
\end{eqnarray}
where $U$ is the utility function of the investor and  ${\mathcal X}(x)$ is the set of all admissible wealth processes $X$  with initial wealth $x$, defined as
\begin{eqnarray*}
{\mathcal X}(x)=\{ X \ge 0 \,\, : \,\, X_t=x + \int_{0}^{t} H_u \ud S_u\;\; \text{for $t\in[0,T]$ }\} \ .
\end{eqnarray*}

Before advancing, we  need to define the set of equivalent supermartingale measures and the set of equivalent local martingale measures

\begin{definition}\label{noarb_hyp}
Let $P\in {\mathcal P}$. We define  
$${\mathcal S}_P= \{Q \sim P \,\, :\,\,  X \mbox{ is $Q$-supermartingale  for all} \,\, X\in{\mathcal X}(x) \}$$ the set of supermartingale measures on  $(\Omega, {\mathcal F}, ({\mathcal F}_{t})_{t\in[0,T]})$ that are equivalent to $P$. We also set ${\mathcal S_P}=\cup_{P\in {\mathcal P}}{\mathcal S}_P$. 
\end{definition}
\begin{remark}
In the statements and arguments that will follow one can easily check that the above defined sets ${\mathcal S}_P$ and ${\mathcal S}_{\mathcal{P}}$ can as well be replaced by the sets ${\mathcal S}^T_P= \{Q \sim P \,\, :\,\,  \E_Q(X_T)\leq x \,\, \mbox{for all} \,\, X\in{\mathcal X}(x) \}$ and $ {\mathcal S}^T_{\mathcal P}=\cup_{P\in {\mathcal P}}{\mathcal S}^T_P$, respectively.
\end{remark}

\begin{definition}
Let $P\in {\mathcal P}$. We define  
$${\mathcal M}_P=\{Q \sim P \,\, :\,\,  Q \, \mbox{ is a local martingale measure}\}$$
 the set of local martingale measures  on  $(\Omega, {\mathcal F}, ({\mathcal F}_{t})_{t\in[0,T]})$ that are equivalent to $P$. We also set ${\mathcal M_P}=\cup_{P\in {\mathcal P}}{\mathcal M}_P$.
\end{definition}

It is clear that ${\mathcal M}_P\subset {\mathcal S}_P$ since all admissible processes are uniformly bounded from below. In what concerns the utility function in the minimax optimization problem \eqref{MIN-MAX-A-A}, we impose that:

\begin{assumption} \label{U_hyp}
$U:\R_{+} \to \R$ is a strictly concave, continuously differentiable, strictly increasing function satisfying the Inada conditions $\lim_{x \to 0^{+}} U'(x)=\infty,\,\, \lim_{x \to \infty}U'(x)=0$ and  the asymptotic elasticity  inequality
\begin{equation*}
AE(U):=\limsup_{x \to \infty} \frac{x U'(x)}{U(x)}<1 \ .
\end{equation*}
\end{assumption}

Furthermore, one needs to impose sufficient hypotheses  on the set of priors ${\mathcal P}$  such that a saddle point for problem \eqref{MIN-MAX-A-A} exists. 

The fundamental condition is:

\begin{assumption}\label{P_hyp}
The set of priors ${\mathcal P}$  is convex and weakly compact.
\end{assumption}

This fundamental assumption may have  to be complemented with further technical conditions which are dependent on the specific choice of model. For example, in a non-dominated continuous semimartingale setting one can adopt the additional  \cite{denis2013optimal}
condition concerning the existence of  a  non empty set ${\mathcal P}_{0}$ of  orthogonal martingale laws satisfying H\"older continuity conditions, such that (i)  for any $P \in {\mathcal P}$,   there exists $P_{0} \in {\mathcal P}_{0}$ with  ${\mathbb E}_{P_{0}}\left[ \left(\ud P/\ud P_{0}\right)^2\right] \le C$ for some constant $C$, and (ii) for any $P_{0} \in {\mathcal P}_{0}$, there exists a $P \in {\mathcal P}$ such that $P \sim P_{0}$ (see Hypothesis (H)  in \cite{denis2013optimal}). 
In the context of  L\'evy models and the particular case of log or power utilities one can adopt the additional  \cite{neufeld2015robust} condition that the uncertainty about drift, volatility and jumps is parametrized by a non empty set $\Theta \subset \R^{d} \times {\mathbb S}_{+}^{d}\times {\mathcal L}$ where ${\mathcal L}$ is the set of L\'evy measures on $\R^{d}$, which is convex and satisfies specific boundedness conditions (see Assumption 2.1 in  \cite{neufeld2015robust}).

Finally, we impose a non arbitrage assumption

\begin{assumption}\label{P_hyp_2}
For each $P\in{\mathcal P}$ we assume that ${\mathcal M}_P\ne \emptyset$. 
\end{assumption}

It should also be made clear that we do not require  the members of the set $\mathcal P$ to be mutually equivalent or even dominated.

\vspace{3mm} 

We will show that, under the assumptions listed above, an investor with preferences described by a Gilboa-Schmeidler minimax utility will place all of his wealth on the riskless asset if and only if the set of priors ${\mathcal P}$ contains a supermartingale measure (i.e. ${\mathcal P}\cap {\mathcal S}\ne\emptyset$). Furthermore, we will show that, under a certain attainability condition (which holds when $\Omega$ is finite or the market is complete), the investor will invest only in the risk free asset if an only if the set of priors $\mathcal P$ contains an equivalent  martingale measure (i.e. ${\mathcal P} \cap \mathcal M \ne \emptyset$). As a first step towards this goal, we will show that the same conclusion holds whenever ${\mathcal P}=\{P\}$ is a singleton, corresponding to the case of Von Neumann-Morgernstern utilities.

\section{An auxiliary result: the case of Von Neumann-Morgernstern utilities}\label{Sec3}

In the case of Von Neumann-Morgernstern utilities, the minimax problem \eqref{MIN-MAX-A-A} reduces to the optimization problem
\begin{eqnarray}\label{u_def}
u(x):=\sup_{X \in {\mathcal X}(x)}  {\mathbb E}_{P}[U(X_T)] \ .
\end{eqnarray}
This problem has been studied in the seminal work of \cite{kramkov1999asymptotic} using duality techniques. As in \cite{kramkov1999asymptotic}, we also assume that  $u(x)<\infty$ for some $x>0$. They have solved \eqref{u_def} in termos of the dual problem 
\begin{eqnarray}\label{v_def}
v(y)=\inf_{Y \in {\mathcal Y}(y)}{\mathbb E}_{P}[V(Y_T)] \ , 
\end{eqnarray}
where ${\mathcal Y}(y)$ is defined as
\begin{eqnarray*}
{\mathcal Y}(y)=\{ Y=(Y_{t})_{t \in [0,T]}  \ge 0 \,\, : \,\, Y_0=y, \,\, X Y=(X_t Y_t)_{t\in[0,T]} \,\,\\
~~~~~~~~~~~~~~~~ \text{is a supermartingale for all $X \in {\mathcal X}(1)$}\} \ ,
\end{eqnarray*}
and $V$ is the Fenchel-Legendre conjugate of the utility function $U$, defined by 
$$
V(y)=\sup_{x >0} \{U(x)-y x\}  \ .
$$

Under the assumptions stated in Section \ref{Sec2}, and in particular under the assumption that  $AE(U)<1$, both the primal and the dual problem admit unique solutions, $\hat{X}:=(\hat{X}_{t})_{t\in[0,T]}\in{\mathcal X}(x)$ and $\hat{Y}:=(\hat{Y}_{t})_{t\in[0,T]}\in{\mathcal Y}(y)$ respectively, related through the identity
\begin{eqnarray}\label{26-7-2016}
\hat{X}_{T}=I(\hat{Y}_{T}) \,\,\, \mbox{for} \,\,\, y=u'(x),  \,\,\, I=(U')^{-1},
\end{eqnarray} whereas   \cite[Thm 2.2]{kramkov1999asymptotic}
\begin{eqnarray}\label{26-7-2016-A}
\hat{X}\hat{Y}=\{ \hat{X}_{t} \hat{Y}_t, \,\,\, t \in [0,T]\}, \,\,\,  \mbox{uniformly integrable martingale}.
\end{eqnarray}
Furthermore, it holds that
\begin{eqnarray}\label{A-123-A}
v(y)=\inf_{Q \in {\mathcal M}}{\mathbb E}_{P} \left[  V\left( y \frac{\ud Q}{\ud P} \right)\right],
\end{eqnarray}
where $\ud Q/\ud P$ denotes the Radon-Nikodym derivative of $Q$ with respect to $P$ on $(\Omega, {\mathcal F}_{T})$, and ${\mathcal M}={\mathcal M}_{P}$ is the set of equivalent local martingale measures for the securities market under consideration.  

We should stress here that the infimum in \eqref{A-123-A} may or may  not be attained. For example, the infimum is always attained if $\Omega$ is a finite probability space (see e.g. \cite[Ch. 3]{delbaen2006mathematics}), or if the market is complete. On the other hand, if $\Omega$ is not finite and the market is incomplete, one can find counter examples as detailed in \cite{kramkov1999asymptotic}.

\begin{proposition}\label{PROP-BASIC}  Suppose  that assumptions \ref{U_hyp}, \ref{P_hyp_2} hold and that ${\mathcal P}=\{P\}$ is a singleton. 
 Then:
\begin{itemize}
\item[(i)] The optimal portfolio is non-random if and only if  $P \in {\mathcal S}$ .
\item[(ii)] If moreover the infimum in \eqref{A-123-A} is attained, then the optimal portfolio is non-random if and only if $P \in {\mathcal M}$ .
\end{itemize}
\end{proposition}

\begin{proof}

(i) If  $P \in {\mathcal S}$, then clearly $\E_P(X_T)\leq x$ for any admissible portfolio. Jensen's inequality and the fact that $U$ is increasing readily imply that ${\mathbb E}_{P}[U(X_T)] \leq U({\mathbb E}_P(X_T))\leq  U(x)$ and thus the non random portfolio of constant value $x$ is the maximizer.

For the converse assume that $\hat{X}=x$  is the maximizer. From \eqref{26-7-2016-A} we have that $\hat{X}_t \hat{Y}_t={\mathbb E}_{P}[\hat{X}_T\hat{Y}_T \mid {\mathcal F}_{t}]$ for any $t \in [0,T]$ which implies that $\hat{Y}_t={\mathbb E}_{P}[\hat{Y}_{T} \mid {\mathcal F}_{t}]$ since $\hat{X}=x$ by assumption. This in turn implies that  $\hat{Y}_{t}=U'(x)$ for every $t \in [0,T]$, since by \eqref{26-7-2016} $\hat{Y}_{T}=U'(x)$ which is also non random. Being optimal, $\hat{Y}=U'(x) \in {\mathcal Y}(y)$ and thus the definition of ${\mathcal Y}(y)$ implies that $XU'(x)$ is supermartingale for all $X\in {\mathcal X}(1)$. Therefore, since $U'(x)>0$, we conclude that any arbitrary admissible process  $X\in {\mathcal X}(x)$ is a $P$-supermartingale, which means that $P\in{\mathcal S}$.

(ii) If $P \in {\mathcal M}$ then $P\in {\mathcal S}$, since the admissible processes $X$ are bounded from below and the proof follows as in case (i) above.

For the converse, suppose that the non random portfolio is the maximizer, i.e.  $\hat{X}=x$.  Then,  the solution of the dual problem is given by $\hat{Y}=y$ (indeed,  we can argue again as in part (i) of the proof to show that if $\hat{X}=x$ then $\hat{Y}$ is constant and $\hat{Y} \in {\mathcal Y}(y)$). Then, \eqref{v_def} implies that 
\begin{eqnarray}\label{C-123-C}
v(y)={\mathbb E}_P(V(\hat{Y}_T))={\mathbb E}_P(V(y))=V(y) \ .
\end{eqnarray}

On the other hand, under the assumption  that the infimum in \eqref{A-123-A} is attained, there exists a measure $Q^{*} \in {\mathcal M}$ such that
\begin{eqnarray}\label{B-123-B}
v(y)={\mathbb E}_{P} \left[  V\left( y \frac{\ud Q^{*}}{\ud P} \right)\right].
\end{eqnarray}
However, since ${\mathbb E}_{P}[\ud Q^{*}/\ud P]=1$, we also have that
\begin{eqnarray}\label{D-123-D}
V(y)= V \left( {\mathbb E}_{P} \left[ y \frac{\ud Q^{*}}{\ud P} \right]  \right)\ .
\end{eqnarray}

Combining \eqref{C-123-C}, \eqref{B-123-B}  and \eqref{D-123-D} we deduce that we have a Jensen's equality with regard to the strictly convex function $V$
\begin{eqnarray*}
{\mathbb E}_{P} \left[  V\left( y \frac{\ud Q^{*}}{\ud P} \right)\right] = V \left( {\mathbb E}_{P} \left[ y \frac{\ud Q^{*}}{\ud P} \right]  \right) ,
\end{eqnarray*}
and, therefore, the random variable $ \frac{\ud Q^{*}}{\ud P} $ is a.e. constant with regard to $P$ and thus $P=Q^*\in{\mathcal M}$.
\end{proof}

\begin{remark}
As one of the referees pointed out, the result of Proposition  \ref{PROP-BASIC} can be rephrased as: The non random portfolio is optimal if and only if $1 \in {\mathcal Y}(1)$, from which one may show   Proposition  \ref{PROP-BASIC}  using essentially the same arguments.
\end{remark}

\begin{remark}
In market models where the infimum in \eqref{A-123-A} is attained for any probability measure $P$ such that ${\mathcal M}_P\ne\emptyset$, one is also tempted to think along the following lines: If $Q\in {\mathcal S}_P$ is a supermartingale measure, then any investor who would believe that the market is governed by $Q$, would choose the non random portfolio, according to Proposition \ref{PROP-BASIC} (i). But then Proposition \ref{PROP-BASIC} (ii) implies that $Q\in {\mathcal M}_P$. Therefore, one may draw the conclusion that ${\mathcal S}_P\subset{\mathcal M}_P$  and since the inverse inclusion holds as well we should have that ${\mathcal S}_P={\mathcal M}_P$. Taking this thought one step further, one wonders with regard to the infimum in \eqref{A-123-A}  whether it is true 
in general true that $\inf_{Q \in {\mathcal M}_P}{\mathbb E}_{P} \left[  V\left( y \frac{\ud Q}{\ud P} \right)\right] = \inf_{Q \in {\mathcal S}_P}{\mathbb E}_{P} \left[  V\left( y \frac{\ud Q}{\ud P} \right)\right]$.
\end{remark}

\begin{remark}
As we have already stated, Proposition \ref{PROP-BASIC} holds in particular for the  case of models in finite probability spaces, and for both single period and multiple periods models. In such cases the  infimum in \eqref{A-123-A} is always attained (see \cite[Theorem 3.2.1]{delbaen2006mathematics}). It should be clear that in such a  setting a simpler proof of our result may be obtained by treating directly the portfolio optimization problem and the resulting variational inequality.  
\end{remark}

\begin{example}
For CARA utilities of the form 
\begin{equation*}
U(x)=\frac{1}{\alpha} x^{\alpha} \ , \qquad \alpha <1 \ ,
\end{equation*}
one can easily compute the Fenchel-Legendre conjugate to obtain
\begin{equation*}
V(y)= -\frac{1}{\nu} y^{\nu} \ ,  \text{with $\displaystyle \nu=\frac{\alpha}{\alpha -1}$ .}
\end{equation*}
In this case, assuming that the infimum is attained for a measure $Q^{*} \in {\mathcal M}$ such that $\ud Q^{*}/\ud P=\phi > 0$ and, additionally, that $P$ is absolutely continuous with respect to Lebesgue measure with density $f \ge 0$, then identity \eqref{C-123-C} assumes the form
\begin{eqnarray*}
1=\int_{0}^{\infty}  \phi(s)^{\nu} f(s) \ud s \ ,
\end{eqnarray*}
where it also holds that
\begin{eqnarray*}
\int_{0}^{\infty} f(s) \ud s = \int_{0}^{\infty} \phi(s) f(s) \ud s=1 \ .
\end{eqnarray*}
As a consequence of H\"older's inequality, the only function $\phi$ with this property is the constant function. Indeed, we can write
\begin{eqnarray*}
1= \int_{0}^{\infty} f(s) \ud s &=& \int_{0}^{\infty} f^{\frac{\nu}{\nu-1}} \phi^{\frac{\nu}{\nu-1}} f^{-\frac{1}{\nu-1}} \phi^{-\frac{\nu}{\nu-1}} \ud s \\
&\le& \left\{ \int_{0}^{\infty} (  f^{\frac{\nu}{\nu-1}} \phi^{\frac{\nu}{\nu-1}} )^{\frac{\nu-1}{\nu}} \ud s \right\}^{\frac{\nu}{\nu-1}} \,
 \left\{ \int_{0}^{\infty} (  f^{\frac{1}{1-\nu}} \phi^{\frac{\nu}{1-\nu}} )^{1-\nu}  \ud s\right\}^{\frac{1}{1-\nu}} \\
&=& \left\{ \int_{0}^{\infty}  f(s) \phi(s) \ud s \right\}^{\frac{\nu}{\nu-1}} \,
 \left\{ \int_{0}^{\infty}  f(s) \phi(s)^{\nu} \ud s\right\}^{\frac{1}{1-\nu}} =1.
\end{eqnarray*}
Since the equality is achieved in H\"older's inequality, there exists a constant $C$ such that 
\begin{equation*}
(f^{\frac{\nu}{\nu-1}} \phi^{\frac{\nu}{\nu-1}})^{\frac{\nu-1}{\nu}}= C (  f^{\frac{1}{1-\nu}} \phi^{\frac{\nu}{1-\nu}} )^{1-\nu}
\end{equation*}
from which it follows that $\phi$ is a constant function and, thus, equal to $1$. Therefore, we obtain that $P=Q^{*}$, as required.
\end{example}

\section{Gilboa-Schmeidler minimax utilities}\label{Sec4}

We now proceed to treat the general case where ${\mathcal P}$ is not a singleton.

\begin{theorem}\label{MAIN-THEOREM-CONTINUOUS}
Suppose that assumptions \ref{U_hyp}, \ref{P_hyp} and \ref{P_hyp_2} hold for the utility function $U$ and the set of priors ${\mathcal P}$.
Consider agents reporting minimax utility of the form 
\begin{equation*}
{\mathcal U}(X)=\inf_{P \in {\mathcal P}} {\mathbb E}_{P}[U(X)]
\end{equation*}
for  random wealth $X$.  The solution to the robust final wealth optimization problem 
\begin{equation*}
\sup_{X \in {\mathcal X}(x)} \inf_{P \in {\mathcal P}} {\mathbb E}_{P}[U(X_T)]
\end{equation*}
is a non-random portfolio if and only if  ${\mathcal P} \cap {\mathcal S} \ne \emptyset$. Furthermore, if the infimum in \eqref{A-123-A} is attainable for every $P\in{\mathcal P}$ then the solution to the optimization problem is the non random portfolio if and only if  ${\mathcal P} \cap {\mathcal M} \ne \emptyset$. 
\end{theorem}

\begin{proof}  Throughout the proof we will use the notation $X^{0}=\{X^{0}_{t}\}_{t \in [0,T]}$ for the wealth process of the non random portfolio that has constant value  $X^{0}_{t}=x$ for every $t \in [0,T]$. 

Assume first that ${\mathcal P} \cap {\mathcal S} \ne \emptyset$ and that $Q \in {\mathcal P} \cap {\mathcal S}$. Applying Jensen's inequality to the concave function $U$ and using the fact that $Q \in {\mathcal S}$  we obtain that for any $X \in {\mathcal X}(x)$, we must have that
\begin{eqnarray*}
U(x) \ge U({\mathbb E}_{Q}[X_T]) \ge {\mathbb E}_{Q}[U(X_T)] \ge \inf_{P \in {\mathcal P}} {\mathbb E}_{P}[U(X_T)] \ ,
\end{eqnarray*}
where the last inequality follows from the fact that $Q \in {\mathcal P}$. Taking the supremum over all $X \in {\mathcal X}(x)$ in the above we conclude that
\begin{eqnarray*}
U(x) \ge \sup_{X \in {\mathcal X}(x)} \inf_{P \in {\mathcal P}} {\mathbb E}_{P}[U(X)] \ .
\end{eqnarray*}
For the non-random portfolio $X^{0}(x)$  the equality is attained, hence $X^{0}$  is a maximizer. Since $\mathcal M\subset \mathcal S$ we would have clearly obtained the same result  if we had started from the assumption ${\mathcal P} \cap {\mathcal M} \ne \emptyset$. 

For the converse, assume that  $X^{0}$ is a maximizer.  Then,
\begin{eqnarray}\label{2-3-2016}
u(x)=\sup_{X \in {\mathcal X}(x)} \inf_{P \in {\mathcal P}} {\mathbb E}_{P}[U(X_T)]=\inf_{P \in {\mathcal P}}{\mathbb E}_{P}[U(X_{T}^{0})]= U(x) \ ,
\end{eqnarray}
since 
$$
{\mathbb E}_{P}[U(X_{T}^{0})]=U(x)
$$
for every $P \in {\mathcal P}$.

 By the saddle point property (\cite[Thm. 1]{denis2013optimal} or \cite[Thm. 2.4]{neufeld2015robust}), it holds that
\begin{eqnarray*}
\sup_{X \in {\mathcal X}(x)} \inf_{P \in {\mathcal P}} {\mathbb E}_{P}[U(X_T)] =\inf_{P \in {\mathcal P}} \sup_{X \in {\mathcal X}(x)} {\mathbb E}_{P}[U(X_T)] \ ,
\end{eqnarray*}
and hence, by \eqref{2-3-2016}, we have that
\begin{eqnarray}\label{2-3-2016-A}
U(x)=\inf_{P \in {\mathcal P}} \sup_{X \in {\mathcal X}(x)} {\mathbb E}_{P}[U(X_T)] \ .
\end{eqnarray}
Since ${\mathcal P}$ is weakly compact, by the lopsided minimax theorem of Aubin and Ekeland \cite[Ch. 6, Sec. 2, Thm. 7]{aubin2006applied} (see also \cite[Lem. 9]{denis2013optimal}), there exists $\hat{P} \in {\mathcal P}$ for which
\begin{eqnarray}\label{2-3-2016-B}
\inf_{P \in {\mathcal P}} \sup_{X \in {\mathcal X}(x)} {\mathbb E}_{P}[U(X_T)] =\sup_{X \in {\mathcal X}(x)} {\mathbb E}_{\hat{P}}[U(X_T)] \ .
\end{eqnarray}
Combining \eqref{2-3-2016-A} with \eqref{2-3-2016-B} we conclude that
\begin{eqnarray*}
U(x)=\sup_{X \in {\mathcal X}(x)} {\mathbb E}_{\hat{P}}[U(X_T)] \ .
\end{eqnarray*}
We have thus reduced the robust problem to the single prior  problem of the previous section with $P=\hat{P}$. 
Hence, by applying  Proposition \ref{PROP-BASIC} we obtain the desired result.
\end{proof}

In the result above, we have considered the market and the prices as given, and the problem of optimal portfolio selection of an uncertainty averse investor with a set of priors concerning the market was studied. It turns out that Theorem  \ref{MAIN-THEOREM-CONTINUOUS} allows one to relate to  the well known results of \cite{dow1992uncertainty}, concerning the effects of uncertainty on the net demand of risky assets, and thus contribute to a better understanding of the phenomenon of the existence of market freezes, which refers to situations where the market endogenously stops as the following example shows.

\begin{example} Consider an one period market starting at $t=0$ and ending at $t=T$ and an agent contemplating positioning on a set of risky assets with payoffs $A=(A_1,\ldots, A_N)$ at time $T$. The agent  reports  a minimax utility with a set of priors ${\mathcal P}$  concerning  the random variable $A$.  Then, the no betting set ${\mathcal N}$, consisting of those asset prices for which the net demand of the assets is zero, is  the convex set ${\mathcal N}=\{ {\mathbb E}_{P}[A] \,\, : \,\, P \in {\mathcal P}\}$.

Indeed, let  $\pi \in {\mathcal N}$. This means that the agent will not take up a  position in this market, hence by Theorem  \ref{MAIN-THEOREM-CONTINUOUS}, restricted in the one period case, this   is equivalent to the  existence of some  $P \in {\mathcal P}$ such that $\pi = {\mathbb E}_{P}[A]$. 

In the special case where the number of risky assets is $N=1$, the no betting set ${\mathcal N}$ is an interval. For instance, Dow and Werlang provide a simple illustrative example in section 2 of \cite{dow1992uncertainty},  where two possible outcomes H and L have respective non additive probabilities $\pi$ and $\pi '$. This is equivalent to considering the set of additive probabilities $(q,\, 1-q)$ corresponding to $(H,\, L)$ where $q$ ranges from $\pi$ to $1-\pi '$. Then they prove that the interval of no betting prices ranges between $\pi H+(1-\pi )L$ and $(1-\pi ' )H+\pi ' L$ which obviously coincides with the interval of  prices that our generalized methodology suggests.
\end{example}

\section*{Acknowledgments}
The authors wish to thank the two anonymous referees for their constructive remarks that motivated us to significantly improve this paper. We also acknowledge useful discussions with Professors A. Tsekrekos and F. Santambrogio. \\
D. Pinheiro research was supported by the PSC-Cuny research award TRADA-46-251, jointly funded by the Professional Staff Congress and the City University of New York.

\bibliography{biblio}
\bibliographystyle{apalike} 

\end{document}